\documentclass[11pt]{amsart}

\usepackage[mathscr]{eucal}
\usepackage{amsmath,amssymb,amscd,amsthm}
\usepackage{color}
\usepackage{epsfig}
\newtheorem{Theorem}{Theorem}
\newtheorem{Lemma}{Lemma}

\newtheorem{Corollary}{Corollary}

\newtheorem{Proposition}{Proposition}
\theoremstyle{definition}

\theoremstyle{plane}

\def \beq{ \begin{equation} }
\def \eeq{\end{equation}}

\title{On a smooth scalar field characterizing \\ the risk  of exposure to methyl-mercury due \\ to non intentional
consumption \\ of shark meet in  males of M\'exico City's metropolitan area}

\begin{document}

\maketitle

\markboth{Elizalde Ram{\'\i}rez, Flores Garc{\'\i}a, Ram{\'\i}rez Romero and  Reyes-Victoria}{Risk characterization of exposure to methyl-mercury}

\vspace{-0.5cm}

\author{
\begin{center}
{\rm LAURA ELIZALDE RAM\'IREZ \\
        Universidad Aut\'onoma Metropolitana Iztapalapa \\
 Departamento de Hidrobiolog{\'\i}a  \\
    Unidad Iztapalapa, CP 09340,
 Cd. de M\'exico,   M\'exico \\
         {\tt laura-eli51@yahoo.com.mx}}\\
                \bigskip
   {\rm EDSON MISSAEL FLORES GARC\'IA, \\
            Universidad Aut\'onoma Metropolitana \\
 Departamento de Matem\'aticas \\
        Unidad Iztapalapa, CP 09340,
 Cd. de M\'exico,   M\'exico \\
         {\tt flores.edson@hotmail.com}} \\
         \bigskip
     {\rm PATRICIA RAM\'IREZ ROMERO \\
        Universidad Aut\'onoma Metropolitana Iztapalapa \\
 Departamento de Hidrobiolog{\'\i}a  \\
    Unidad Iztapalapa, CP 09340,
 Cd. de M\'exico,   M\'exico \\
         {\tt pattdf@gmail.com}} \\
                \bigskip
          {\rm J. GUADALUPE REYES VICTORIA \\
            Universidad Aut\'onoma Metropolitana \\
 Departamento de Matem\'aticas \\
        Unidad Iztapalapa, CP 09340,
 Cd. de M\'exico,   M\'exico \\
         {\tt revg@xanum.uam.mx}}
\end{center}}

 \bigskip
\begin{center}
\today
\end{center}

\begin{abstract}
In this article, we obtain, through statistical and numerical
methods, a smooth function in the variables of the life stage and the
concentration, which estimates the risk of exposure of methylmercury
due to the unintentional consumption of shark in men from Mexico City.
With methods of the Theory of Singularities and Dynamical Systems,
the stability of this risk function was shown by analyzing the associated
vector field. The region of risk was obtained in the variables as mentioned above, and the average risk in the whole region was calculated, which turns out to be a high index. The associated risk surface is a Hadamard surface embedded in the three-dimensional space $\mathbb{R}^3$, and the points where the curvature is zero determine critical ages important for the risk in men.
\end{abstract}

\footnote*{MSC: Primary XXX, Secondary XXX}

\smallskip

\footnote*{Keywords: Risk smooth function, Morse stable function,  Risk vector field, Hadamard surface.}

\section{Introduction}\label{sec:intro}

We continue here the study begun by Elizalde {\it et al.} in (Elizalde, \cite{Elizalde 2}) of characterizing the risk  of exposure to methyl-mercury due  to non intentional  consumption  of shark meet in  of M\'exico City's metropolitan area, there for females and here for males.

As is shown in (\cite{Elizalde 2}), studies in Mexico have reported methyl-mercury concentrations in commercial sharks from 0.27 to 3.33 ppm, but also, the potential substitution of fish meat with shark meat in edible products (Ram{\'\i}rez-Romero {\it et al.}, \cite{Ramirez}). This happens because there are no morphological features that can help to differentiate sharks meat from other products, since the former are purchased processed (without fins or head, as fillet, smoked or as ground meat). 104 shark species has been reported in Mexico; of these, 55 inhabit the Pacific and the rest are distributed in the Gulf of Mexico and the Caribbean (Espinosa-P\'erez {\it et al.}, \cite{Espinosa}). Their economic importance is based on the use of its meat and fins. One of the main distribution sites,and the biggest wholesale market for these products in Mexico City (CDMX), is the Central de Abasto de Pescados y Mariscos, here fish meat is bought fresh, ground, dried-salted (like cod), smoked, in pieces or in fillets.To find out the frequency of the substitution of fish meat for shark meat, Elizalde (2018) bought different fish presentations in the Fish and Seafood market of Central de Abasto, and analyzed their identity through Polymerase Chain Reaction (PCR) using universal shark oligonucleotides. Fifty-three samples were analyzed,  of which 60.37 $\%$ were positive for the replacement of shark species.
A lower limit of the reference dose of 95 $\% $ has been selected at an effect level of the obtained 5 $\%$ . Application of a hazard coefficient (RQ) power model ($RQ> 1$) to response data based on previous studies conducted in the Faroe Islands (USEPA, \cite{USEPA2}). The risk quotient (RQ) calculation indicated a high risk for the analyzed population. Therefore, the objective of this study was to calculate the health risk for males of the Mexico City metropolitan area from unintentional exposure to methylmercury through the consumption of shark.
A less biased approach to risk assessment uses uncertainties analysis to assess the degree of confidence that can be given to risk estimation. But nevertheless, meat mathematically when studying some phenomenon of nature, such as pollution and the effects of this cause on biota and man, it must be understood that observations are subject to errors. Trying to formalize this phenomenon by means of a formula, should be understood as an approximation of reality using numerical and qualitative methods (Reyes, \cite{Reyes})

\section{Methods}\label{sec:methods}

\subsection{Data acquiring}\label{subsec:data-adquiring}
The analyzed samples were collected from Mexico City’s Central de Abasto. The sampled products included “fish” meat to make ceviche, meat to make fish broth, meat to make fish quesadillas, smoked fillet, inexpensive steak (sea bass, Nile fish white fillet, catfish, etc.) and breaded fillet.

Positive control shark samples (Carcharhinus. limbatus, Carcharhinusleucas, Carcharhinusfalciformis, Galeocerdocuvie, Isurusoxyrinchus) were donated by the UNAM Genetics Laboratory. Negative control samples were fish from different species: red snapper (Lutjanuscampechanus), marlin (Istiophoridaesp), catfish (Siluriformes), Nile fish (Oreochromismossambicus), sea bass (Centropomusundecimalis) and salmon (Oncorhynchussp).

\subsection{Survey design}\label{subsec:survey}

A non-probabilistic sampling was done, also called discretionary sampling (M\'endez, \cite{Mendez}), in order to identify the population characteristics and consumption habits (quantity and frequency of fish products consumption). The surveys application sites were selected markets in various municipalities of Mexico City, such as Iztapalapa, Xochimilco, Iztacalco, Coyoac\'an, and Benito Ju\'arez; and in municipalities of the State of Mexico (Ixtapaluca and Nezahualcoyotl). The information collected was the frequency of fish consumption, portion size (weight in grams), species and presentation type; age, gender and weight of the respondent and his entire family. A total of 777 surveys were applied.

\subsection{Modeling the dose}\label{subsec:mod-dose}

The average daily dose during the lifetime (LADD) or the daily chronic ingestion (CDI, chronic daily intake) is a function of the average concentrationof the contaminant and the ingestion rate (by oral route). The parameters used (body weight, age, sex, consuming preferences andfrequency) were obtained from the aforementioned surveys. Additionally, the average life expectancy of Mexican male consumers (78 years) was obtained from national statistics available on line (INEGI, \cite{INEGI}).
The total dose and the average daily dose (ADD) were calculated with following equations (Evans {\it et al.}, \cite{Evans}):
\begin{equation}
\rm{Total \, \,  dose} = \rm{(concentration) (ingestion) (duration) (frequency) }
\end{equation}
\begin{eqnarray}
\rm{Average \, \, Daily \, \, Dose} &=& \rm{(Total  \, dose) / (Body  \, weight  \times Life  \,expectancy) } \nonumber \\
&& \rm{ (mg / kg-day) } \nonumber \\
\end{eqnarray}

  The total dose for three concentrations of methylmercury were calculated; these were obtained from Ram{\'\i}rez-Romero {\it et al.} \cite{Ramirez} and correspond to the minimum (0.27 mg / Kg), the average (2.43 mg / Kg) and the maximum (3.33 mg / Kg) methylmercury concentrations used in this work.

\subsection{Reference dose}\label{subsec:dose}

The quantitative health risk assessment of a non-carcinogenic agent is based on a reference dose, which is an estimate (with uncertainty spanning an order of magnitude of 10) of a daily exposure where sensitive human subgroups are included. The chosen level was a lower limit of reference dose of 95$\%$ at an effect level of 5 $\%$ obtained by applying a K power model ($K> 1$) to dose response data based on previous studies conducted in the Faroe Islands (USEPA, \cite{USEPA2}). The reference dose is 0.0001 mg / kg / day for men of reproductive age, children and older adults; and, 0.0003 mg / kg / day for the adult population.

\subsection{Calculation of risk for unintentional consumption of shark meat}\label{sec:Calculation}

For the calculation of weekly and monthly consumption, we used the United States Environmental Protection Agency equation (USEPA, \cite{USEPA2}):
\begin{equation}
CR_{mm}= \frac{CR_{lim}  \times T_{ap}}{MS}
\end{equation}
where,
\begin{eqnarray*}
CR_{mm}& =& \rm{ maximum \, \, consumption \, \, allowed  \, \, in  \, \, fish  \, \, portion  \, \, (meals / month)} \nonumber \\
CR_{lim} &=& \rm{Maximum \, \, consumption \, \,  allowed\, \,  in\, \,  fish \, \, portion\, \,  (Kg / day)}  \nonumber \\
T_{ap} & =& \rm{average \, \, period  \, \, 365.25 \, \, days / 12 months\, \,  = \, \, 30.44 \, \, days / month}  \nonumber \\
MS& =& \rm{fish \, \, portion \, \, weight.}  \nonumber  \\
\end{eqnarray*}

This, for the sensitive population made up of men of reproductive age 12 to 50 years old and children under 12 years of age.

On the other hand, to determine the maximum consumption allowed for the sensitive population, in fish portions, in kilograms per day, with information on women
of reproductive age from the aforementioned survey, we use the equation,
\begin{equation}
CR_{lim}=\frac{RfD \times BW}{C_m}
\end{equation}
where,
\begin{eqnarray*}
BW& =& \rm{Consumer \, \,body \, \, weight (Kg)} \nonumber \\
Cm & =& \rm{Concentration \, \,of \, \,mercury \, \,in \, \,fish \, \,species \, \,(mg / Kg)} \nonumber \\
RfD & =& \rm{0.0001 \, \, mg / Kg \, \, reference\, \, dose-day.} \nonumber \\
\end{eqnarray*}
where the reference dose-day is for the developing foetus and men of childbearing age.

We remark that the reference dose is 0.0001 mg / Kg according to the toxicological effects of methylmercury (EPA, \cite{USEPA1}).

For the calculation of weekly and monthly consumption in Kg / day for the adult population, equations (3) and (4) were used, but with the reference dose of 0.0003 mg / kg / day
proposed by the USEPA in 1980, which is based on the methylmercury poisoning in Iraq in 1970, when wheat grain was treated with fungicides with methylmercury,
which was ground and turned into flour for consumption.

\subsection{Health Risk Characterization}\label{subsec:health-risk-car}

For the health risk analysis, the hazard or risk coefficient was calculated with the following  relation (Evans {\it et al.}, \cite{Evans}),
\begin{equation}\label{eq:risk-coeff}
\rm{Risk \, \, coefficient}= \frac{\rm{Exposure}}{\rm{RfD}}
\end{equation}
which is equal to the risk ($R$).
The exposure ($E$), is obtained through the equation (ATDSR, \cite{ATDSR}):
\begin{equation}\label{eq:E-coeff}
E= \frac{C \times TI \times FE}{PC}
\end{equation}
where,
\begin{eqnarray*}
\rm{C} &=& \rm{Concentration \,\, of \,\,  the  \,\, contaminant \,\,  in \,\,  fish \,\,  (mg / Kg / day)} \nonumber \\
\rm{TI} &=& \rm{Intake \,\,  rate \,\,  (mg)} \nonumber \\
\rm{FE} &=& \rm{Exposure \,\,  factor \,\,  (without \,\,  units)} \nonumber \\
\rm{PC} &=& \rm{Body  \,\, weight \,\,  (Kg)} \nonumber \\
\end{eqnarray*}

The exposure factor allows us to calculate the dose of contaminant that is ingested. However, it is compared with the administered dose used in
experimental animal studies designed to obtain the dose-response relationship. The exposure factor was calculated using equation 7 for the different groups,
separated by age of the analyzed population (ATDSR, \cite{ATDSR}). According to Elizalde (\cite{Elizalde 1}) the genetics results showed an average of 60.37 $\%$ substitution
of fish meat for shark meat which was considered in the analysis.
\begin{equation}\label{eq:FE-coeff}
FE=\frac{(\rm{exposure  \,\, in \,\, days/weeks})(52 \,\, \rm{weeks/year})(\rm{exposure\,\,years})}{(\rm{years\,\,exposure})(365 \,\, \rm{days/years)})}
\end{equation}

According to Evans {\it et al.}, (\cite{Evans}) the result of the value for the risk coefficient is interpreted as follows:
\begin{eqnarray*}
R &<& 1  \,\, (\rm{acceptable \,\,  risk}) \nonumber \\
R &>& 1  \,\, (\rm{unacceptable \,\,  risk}) \nonumber \\
\end{eqnarray*}

\section{Results}\label{sec:results}

\subsection{Survey}\label{subsec:survey}

The total number of people included in the survey was 1976, where men consume fish meat more frequently and in greater quantity: 262.60 g / month, followed by men and seniors:194 g / month and 193 g / month respectively (see Table 1).

\begin{center}
\begin{tabular}{|l|c|c|c|c|c|c|c|c|} \hline
\multicolumn{6}{|c|}{\mbox{Table 1. Population characteristics and consumption habits}} \\ \hline
 \mbox{Surveyed}  &  \mbox{NPS} &  \mbox{AG} &  \mbox{ABW} &  \mbox{AIR} &   \mbox{AFP}  \\ \hline
 $\mbox{Children}$   & $ 421 $    & $ 0-14 $ & $34.94 $ & $188.17 $ & $ 1.3$  \\ \hline
  $\mbox{Men}$      & $ 546$   &     $15-59 $    & $73.44 $   & $262.60 $ &  $2.6$    \\ \hline
   $\mbox{Senior}$   & $396 $   & $60-90 $    & $68.85  $   & $193.38 $ &  $2.1  $    \\ \hline
\end{tabular}
\end{center}

\smallskip

In such that Table, NPS is the number of people surveyed, AG is the age (in years), ABW is the average body weight (in Kg),  AIR is the average intake rate (in g) and
AFP is the average Fish portions consumed  per month.

The consumption habits of the analyzed population showed that the most preferred product is the fish fillet (Table 2), followed by fish nuggets population and smoked fish.
\begin{center}
\begin{tabular}{|l|c|c|c|c|c|c|c|c|} \hline
\multicolumn{3}{|c|}{\mbox{Table 2. Fish product preferences of people of Mexico City's Metropolitan Area}} \\ \hline
 \mbox{Product}  &  \mbox{Adult} &  \mbox{Sensible population}   \\ \hline
 $\mbox{Fish fillet (g)}$                         & $ 65 \% $    & $31 \%$   \\ \hline
 $\mbox{Fish Meat for ceviche (g)}$      & $12 \% $   & $16 \% $       \\ \hline
 $\mbox{Fish Meat for fish broth (g)}$   & $10 \% $   & $ 17 \%$       \\ \hline
 $\mbox{Smoked fish (g)}$   & $8 \% $   & $18 \% $       \\ \hline
 $\mbox{Fish nuggets}$   & $ 14 \%$   & $ 18 \%$       \\ \hline
\end{tabular}
\end{center}

\subsection{Dose Modeling}\label{subsec:dose-model}

The average daily dose calculated for the minimum, average and maximum methyl-mercury concentrations are shown in Table 3
of reproductive age are at a minimum consume a dose that does not exceed the reference dose when the minimum MeHg concentration was considered; however,

\begin{center}
\begin{tabular}{|l|c|c|c|c|c|c|c|c|} \hline
\multicolumn{5}{|c|}{\mbox{Table 3. Average and reference of daily methylmercury (MeHg) dose for different age groups.}} \\ \hline
 \mbox{}  &  \mbox{Average daily dose } &  \mbox{} &  \mbox{} &  \mbox{Reference dose}  \\ \hline
 $\mbox{Age group}$   & $\mbox{[0.27 mg/Kg] HgMe} $    & $ \mbox{[2.43mg/Kg] MeHg} $ &  $ \mbox{[3.33 mg/Kg] MeHg} $ & $\mbox{mg/Kg MeHg} $ \\ \hline
  $\mbox{Boys}$      & $0.0003 $   & $0.0023 $    & $ 0.0032$   & $0.0001 $    \\ \hline
  $\mbox{Men}$      & $ 0.0002$   & $0.0015 $    & $0.0021 $   & $0.0003 $    \\ \hline
   $\mbox{Senior Men}$   & $ 0.0001$   & $0.0012 $    & $0.0017 $   & $0.0001 $     \\ \hline
\end{tabular}
\end{center}

\subsection{Analysis of health risks due to unintentional consumption of shark meat}\label{subsec:analysss}

To obtain the maximum allowed number of portions that can be consumed without causing adverse health effects, equations (3) and (4), described in the method (USEPA, \cite{USEPA2}), were used. Taking into account the result of the genetic analysis of the different fish presentations, in which a 60.37 $\%$ substitution of fish meat for shark was obtained (Elizalde, \cite{Elizalde 1}), the maximum consumption allowed for all population groups was recalculated, for the minimum, average and maximum concentrations of methylmercury in shark meat.

\subsection{Risk Coefficient}\label{subsec:risk-coeff}
The health risk for men due to unintentional shark meat consumption for the different age groups, was calculated for the three concentrations; for the calculation of the risk coefficient, equations (\ref{eq:risk-coeff}), (\ref{eq:E-coeff}) and (\ref{eq:FE-coeff}) were used, the results can be seen in Table 4; for the low methylmercury concentration with a 60.37 $\%$ substitution for shark meat, the risk coefficient is less than one, which means, that in general, the unintentional consumption of shark meat does not pose a risk or is an acceptable health risk; however, for children from 0 to 5 years old, the calculated value (0.785) is closer to one, which alerts us to the possible risk that slightly higher may represent. For example with a MeHg of (0.45 mg / kg) the RC exceeds one. The risk coefficient for the medium and high MeHg concentrations was always well above 1, which means that the consumption habits represent a risk for the entire population.
\begin{center}
\begin{tabular}{|l|c|c|c|c|c|c|c|c|} \hline
\multicolumn{4}{|c|}{\mbox{Table 4. Risk coefficient of men's unintentional consumption of shark meat.}} \\ \hline
 \mbox{}  &  \mbox{Risk coefficient} &  \mbox{Risk coefficient} &  \mbox{Risk coefficient}  \\ \hline
 $\mbox{Age group years}$   & $\mbox{[0.27 mg/Kg] HgMe} $    & $ \mbox{[2.43mg/Kg] MeHg} $ &  $ \mbox{[3.33 mg/Kg] MeHg} $ \\ \hline
  $\mbox{Babies $(1- 6)$}$      & $0.804 $   & $7.237 $   & $9.918 $    \\ \hline
  $\mbox{Boys $(6 -12)$}$      & $ 0.342$   & $3.077 $    & $4.216$   \\ \hline
   $\mbox{Men $(12 -60)$}$      & $0.204 $   & $ 1.834$    & $2.513 $     \\ \hline
   $\mbox{Senior $(60-90)$}$   & $0.388 $   & $3.490 $    & $ 4.783 $   \\ \hline
\end{tabular}
\end{center}

\section{The scalar field of risk}\label{subsec:esc-risk}

In this section we find the scalar field which will give us information of the process.

\subsection{The escalar field}\label{subsec:esc-field}

Life stages $[1, 90]$ are conveniently reparametrized so that they adapt to an interval $[1,5]$, this is, if $s \in [1, 90]$ is the real age, we
will use the variable $t \in [1,5]$, and the functional relation is given by
\begin{equation} \label{eq:Reperametrization}
     s(t) = \left\{
	       \begin{array}{ll}
		 5t-4     & \mathrm{if\ }  1 \leq t \le 2 \\
		6t-6    & \mathrm{if\ }  2 \leq t \le 3 \\
		48t-132     & \mathrm{if\ }  3 \leq  t \le 4 \\
		30t-60     & \mathrm{if\ }  4 \leq t \le 5
	       \end{array}
	     \right.
   \end{equation}

This later is because the fish consumption begins after the first year of life. Therefore, with such reparametrization (\ref{eq:Reperametrization}), the intervals of stage age are applied:

\[ \mbox{Babies, $[1, 6)$ years, into the interval $[1,2)$,}\]

\[ \mbox{Boys, $[6,12)$ years into the interval $[2,3)$,}\]

\[ \mbox{Men, $[12,60)$ years into the interval $[3,4)$,}\]

\[ \mbox{Senior men, $[60,90]$ years into the interval $[4,5]$.}\]	

\medskip

With these conditions, and using the Interpolation method (Reyes, \cite{Reyes}) three polynomials of degree 4
are obtained, that soften the polygonal graphs involving the data for the given stage intervals, and for each concentration of methylmercury (MeHg): 0.3, 2.7, 3.7 mg / Kg
respectively, as it is shown in Figures \ref{fig:interpolation1}, \ref{fig:interpolation2} and \ref{fig:interpolation3}.

\begin{figure}[ht]
\centering
\includegraphics{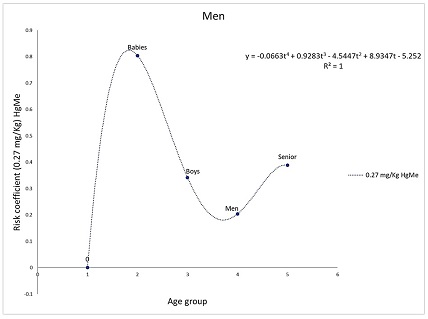}
\caption{Risk Coefficient for men of different life stages considering the concentration 0.27.}
\label{fig:interpolation1}
\end{figure}

\begin{figure}[ht]
\centering
\includegraphics{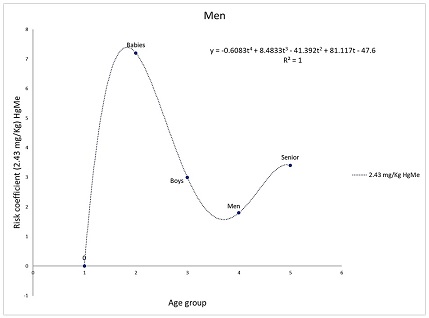}
\caption{Risk Coefficient for men of different life stages considering the concentration 2.43.}
\label{fig:interpolation2}
\end{figure}

\begin{figure}[ht]
\centering
\includegraphics{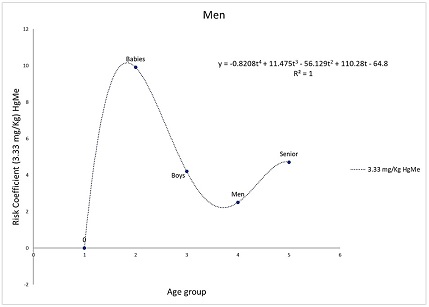}
\caption{Risk Coefficient for men of different life stages considering the concentration 3.43.}
\label{fig:interpolation3}
\end{figure}

This is, in the principal variables $t$ as the stage variable, and $c$ as the concentration variable, we can obtain the following particular polynomials.
\begin{center}
\begin{tabular}{|l|c|c|c|c|c|c|c|c|} \hline
\multicolumn{2}{|c|}{\mbox{Table 3. Interpolating polynomials }} \\ \hline
 \mbox{Concentration}  &  \mbox{Associated polynomials of degree 4 in the variable $t$}   \\ \hline
 $\mbox{[0.27] Hg}$      & $ R_{0.27} (t) =  -0.06 t^4 + 0.92 t^3 - 4.54 t^2 +  8.93t-5.25$      \\ \hline
 $\mbox{[2.43] Hg}$      & $R_{2.43} (t) = -0.60 t^4 + 8.48t^3  - 41.39 t^2 +  81.11 t-47.6$          \\ \hline
 $\mbox{[3.33] Hg}$   & $R_{3.33} (t) =  -0.82 t^4 + 11.47 t^3 - 56.12 t^2 + 110.28 t -64.8$         \\ \hline
\end{tabular}
\end{center}

\smallskip

In order to construct a global function $R(t,c)$  in the stage and concentration variables $(t,c)$ that estimates the risk in the domain $D=[1,5] \times [0.2,3.5]$,
such that for each value of concentration c we have a polynomial relation $R_{c}(t)$ that depends only on stage t, we propose the following:
\begin{equation}\label{proposed function}
R(t,c)=f_{4} (c)  t^4+ f_{3} (c) t^3+f_{2} (c) t^2  + f_{1} (c)  t + f_{0} (c)
\end{equation}
where the functions $ f_{k} (c)$ are obtained by the linear regression method according to the conditions of the obtained polynomials:
\begin{eqnarray*}
f_{4} (0.27) &=& - 0.0663,  \, f_{4} (2.43) = - 0.60, \, f_{4} (3.33) = - 0.82;  \nonumber \\
f_{3} (0.27) &=& 0.92,  \, f_{3} (2.43) = 8.48,  \, f_{3} (3.33) = 11.47; \nonumber \\
f_{2} (0.27) &=&  -4.54,  \, f_{2} (2.43) = - 41.39, \,  f_{2} (3.33) = - 56.12,  \, \mbox{ etc...} \nonumber \\
\end{eqnarray*}

The first function $f_{4} (c)$ in (\ref{proposed function}) is obtained using  the Excel program and it is given by the linear relation,
\begin{equation}
f_{4} (c)= -0.24 c + 0.006
\end{equation}

The other linear functions are obtained in a similar way, obtaining,
\begin{eqnarray}
f_{3} (c) &=& 3.45 c + 0.0076 \nonumber \\
f_{2} (c) &=& -16.89 c -    0.06   \nonumber \\
f_{1} (c) &=& 33.17 c + 0.09  \nonumber \\
f_{0} (c) &=&  -19.48 c - 0.04 \nonumber \\
\end{eqnarray}

In this way, the searched scalar field (\ref{proposed function}) that estimates the risk of methyl-mercury in region $D$ becomes,
\begin{eqnarray}\label{riskfunction}
R (t,c) &=& (-0.24 c + 0.006) t^4 + (3.45 c + 0.007) t^3 + (-16.89 c -    0.06) t^2  \nonumber \\
                  &+&  (33.17 c + 0.09) t + (-19.48 c - 0.04). \nonumber \\
\end{eqnarray}

For any fixed given value of the concentration c the corresponding function $R_{c} (t)$ has a graphic in the plane t, R as it is shown in the interpolation process (see Figure \ref{fig:rick-c-fixed}).
\begin{figure}[ht]
\centering
\includegraphics{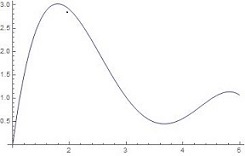}
\caption{Graphic of the function $R_{c}(t)$.}
\label{fig:rick-c-fixed}
\end{figure}

We calculate the  gradient of the risk function $R_{c} (t)$ and obtained,
\begin{eqnarray}\label{eq:risk-vector}
\nabla R(t,c) &=& \left( \frac{\partial R}{\partial t}, \frac{\partial R}{\partial c} \right) \nonumber\\
                    &=& (0.09 + 33.17 c + 2 (-0.06 - 16.89 c) t + 3 (0.007 + 3.45 c) t^2  \nonumber \\
                     &+& 4 (0.006 - 0.24 c) t^3, 19.48 + 33.17 t - 16.89 t^2 + 3.45 t^3 - 0.24 t^4 )  \nonumber \\
\end{eqnarray}

In order of finding the critical points of  (\ref{riskfunction}), we solve the system of algebraic equations in the variables $t,c$,
    \begin{eqnarray}\label{sys:critical}
    0 &=& (0.09 + 33.17 c + 2 (-0.06 - 16.89 c) t + 3 (0.007 + 3.45 c) t^2+ 4 (0.006 - 0.24 c) t^3   \nonumber \\
                     0 &=& 19.48 + 33.17 t - 16.89 t^2 + 3.45 t^3 - 0.24 t^4, \nonumber \\
	 \end{eqnarray}
which, as can be seen, has not solutions into the domain $D$.
\smallskip

We recall that a {\it stable function} $f$ defined in the compact set $D$ is such that  every nearby function $g$ defined in $D$ is identical to $f$ (Golubitsky-Guillemin, (\cite{Golubitsky-Guillemin})).

Also, a {\it Morse function} is such that one with non degenerate critical points with different critical values (Golubitsky-Guillemin, (\cite{Golubitsky-Guillemin})).

Since the risk function $R(t,c)$ has not critical points it follows the following result.
 \begin{Theorem}
  The risk field $R (t, c)$ is a stable Morse function in the simply connected compact  set $D$
 \end{Theorem}
    \begin{proof}
  The whole set $D$ is a regular set for $R (t, c)$ , which proves the Morse property. The stability follows from the Mather-Malgrange Theory (see Proposition 2.2 in \cite{Golubitsky-Guillemin})
   \end{proof}

 Therefore, under small smooth deformations of the risk function $R (t, c)$ in $D$, the deformed function obtained has the same qualitative behaviour.
 In other words, any small error in obtaining the data would lead to a new risk relationship with the same characteristics.

 On the other hand, we define the Critical Risk Region inside the domain $D$ as the subset
\begin{equation}
R(D)=\{(t,c) \, \, | \, \, R (t, c) \geq 1\}
\end{equation}
 and it is represented as a coloured contour in Figure \ref{fig:critic-region}, which shows a high risk region, as expected from the data in Table 4.
\begin{figure}[ht]
\centering
\includegraphics{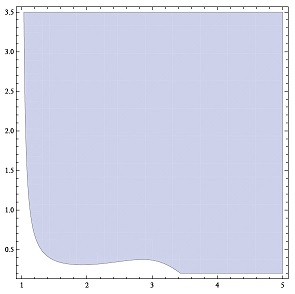}
\caption{Critical risk region $R(D)$.}
\label{fig:critic-region}
\end{figure}

The following result shows the risk probability for the whole process.

\begin{Proposition}
The probabilty of risk of exposure of methylmercury $P$ for the considered stages and concentrations is high.
\end{Proposition}
\begin{proof}
We calculate the ratio between the corresponding areas of $D$ and  $R(D)$, obtaining the aforementioned probablitity of risk,
\begin{eqnarray}
P &=& \frac{\mbox{Area(R(D))}}{\mbox{Area(D)}} \nonumber \\
   &=& \frac{1}{\mbox{Area(D)}}\int  \int _{R(D)}  \, dc \, dt  \nonumber \\
                 &=& \frac{12.92}{13.2}= 0.97   \nonumber \\
\end{eqnarray}
\end{proof}

The contour lines or level curves of the scalar risk field  are displayed in domain $D$ in Figure \ref{fig:level-curves}. The darker region
 indicates less risk, while the lighter region indicates greater risk.
\begin{figure}[ht]
\centering
\includegraphics{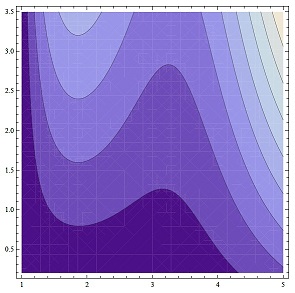}
\caption{Level curves due to unintentional shark consumption for men.}
\label{fig:level-curves}
\end{figure}

We obtain also the following crucial and  important result.

\begin{Theorem}\label{theo:mean-risk}
The average value $R^*=5.55$  of  $R(t,c)$  in the whole set $D$ represents a high risk for the population.
\end{Theorem}
\begin{proof}
The average value $R^*$ of the risk in domain $D$ is calculated by applying the formula (Rudin, \cite{Rudin}),
\begin{eqnarray}
R^* &=& \frac{1}{\mbox{Area(D)}}  \int  \int _{D}  R (t, c) \, dc \, dt  \nonumber \\
       &=& \frac{1}{\mbox{Area(D)}}  \int _{t=1}^{t=5} \int _{c=0.2}^{c=3.5}[(-0.24 c + 0.006) t^4 + (3.45 c + 0.007) t^3     \nonumber \\
       &+ & (-16.89 c -0.06) t^2 + (33.17 c + 0.09) t + (-19.48 c - 0.04)] \, dc \, dt  \nonumber \\
              &=& \frac{73.39}{13.2}=5.55  \nonumber \\
\end{eqnarray}

  Such that number represents a high risk in the whole set D.
\end{proof}

\subsection{The risk vector field}\label{subsec:vec-field}

Because there are not critical point for the risk function, we study the risk gradient vector field  (\ref{eq:risk-vector})  for understanding the behaviour of the risk function. The flow of such risk vector field shows how the process is changing along the solutions of the associated system of differential equations,
 \begin{eqnarray}\label{eq:dynamical system}
    \frac{d t}{d \tau} &=& 0.09 + 33.17 c + 2 (-0.06 - 16.89 c) t + 3 (0.007 + 3.45 c) t^2 + 4 (0.006 - 0.24 c) t^3, \nonumber \\
    \frac{d c}{d \tau} &=& -19.48 + 33.17 t - 16.89 t^2 + 3.45 t^3 - 0.24 t^4, \nonumber \\
	 \end{eqnarray}
where $\tau$ is the dynamic time (see Figure \ref{fig:vector-field}).

\begin{figure}[ht]
\centering
\includegraphics{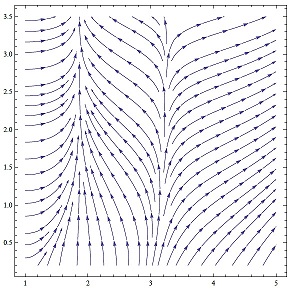}
\caption{Vector field of risk $\nabla R(t,c)$ for unintentional consumption of shark for men.}
\label{fig:vector-field}
\end{figure}

\begin{Lemma}\label{lema:not-periodic}
The dynamical system (\ref{eq:dynamical system}) does not have neither equilibrium points, nor closed orbits in the compact simply connected region $D$ (see Figure \ref{fig:vector-field}).
\end{Lemma}

\begin{proof}
Since the risk function does not have critical points in the considered domain, it follows that there are not equilibrium points in $D$
for such system. From the Poincar\'e-Bendixon Theorem (Guckenheimer-Holmes, \cite{Guckenheimer}) follows that there are not periodic orbits, since in other case, if there is one periodic orbit inside $D$, the simply connected region bounded by this orbit must contain one equilibrium point.
\end{proof}

The same Figure shows vertical lines in the  flow of the vector field of risk and are understood as the ages where there is a significant risk.
These ages will be calculated later using geometric methods.

\subsection{The risk surface}\label{subsec:risk-surface}

The so-called associated {\bf risk surface} $S$ is the graphic of the risk function (\ref{riskfunction}), and it is a two dimensional surface embedded in the three dimensional Euclidean space $\mathbb{R}^3$, shown in Figure \ref{fig:risk-surface}.
\begin{figure}[ht]
\centering
\includegraphics{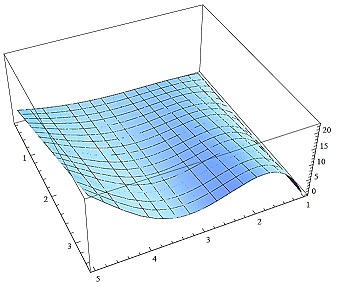}
\caption{Risk surface for men due to unintentional exposure to MeHg}
\label{fig:risk-surface}
\end{figure}

We use the Gaussian curvature function $K (t, c)$ (Dubrovine {\it et al}, \cite{Dubrovin}) of the Risk Surface to determine the critical ages of the global risk
 function.

We recall that one Hadamard surface has non positive Gaussian curvature in all its points.

\smallskip

We have the following important result.

 \begin{Theorem}\label{theo:hadamard-surface}
 The associated  risk Surface S is one Hadamard surface embedded in the three dimensional Euclidean space $\mathbb{R}^3$
  \end{Theorem}

 \begin{proof}
 If we parametrize the surface $S$ on the domain $D$ in the canonical way,
 \begin{equation}\label{eq:parametrization}
 \varphi (t,c) =(t,c,R(t,c)), \quad (t,c) \in D,
 \end{equation}
 the Gaussian curvature is calculated with the equality  (see Dubrovine {\it et al.}, \cite{Dubrovin}),
 \begin{eqnarray}\label{eq:curvature}
 K(t,c) &=&  \frac{\left( \frac{\partial^2 R}{\partial t^2} \right)\left(\frac{\partial^2 R}{\partial c^2}  \right)-\left( \frac{\partial^2 R}{\partial t \, \partial c} \right)^2}{\left( 1+
  \left( \frac{\partial R}{\partial t}\right)^2+ \left(\frac{\partial R}{\partial c} \right)^2   \right)^2} \nonumber \\
   &=& -\frac{\left(33.17 - 33.78 t + 10.35 t^2 - 0.96 t^3 \right)^2}{\left( 1+
  \left( \frac{\partial R}{\partial t}\right)^2+ \left(\frac{\partial R}{\partial c} \right)^2   \right)^2} \nonumber \\
 \end{eqnarray}
  because
\[\frac{\partial^2 R}{\partial c^2}=0\]
 and
\[\frac{\partial^2 R}{\partial t \, \partial c}= 33.17 - 33.78 t + 10.35 t^2 - 0.96 t^3\]
 in the whole set $D$.

 Therefore, $S$ has a non positive curvature and consequently it is one Hadamard surface. It is also embedded in the three dimensional space because it is the graphic of the smooth  risk function. This ends the proof.
 \end{proof}

 A consequence of this result is the important following result for the process.

  \begin{Corollary}\label{cor:risk-ages} The critical ages for the process are, in the biological time $s$,
 \begin{equation}
 s= 5 \,  ( \mbox{years}),   \quad s=26.4 \, ( \mbox{years}), \quad  s=105 \,  (\mbox{years})
 \end{equation}
  \end{Corollary}

 \begin{proof}
  The expression of the curvature $K(t,c)$ in (\ref{eq:curvature})  shows the sign of such a curvature function is completely determined on $D$ by the reduced function
  \begin{equation}
  k(t,c) = -(33.17 - 33.78 t + 10.35 t^2 - 0.96 t^3)^2
 \end{equation}

The graphic of k(c,t) is shown in Figure \ref{fig:curvature}.
\begin{figure}[ht]
\centering
\includegraphics{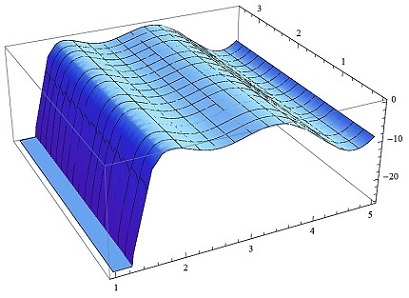}
\caption{Curvature of the Risk Surface for men}
\label{fig:curvature}
\end{figure}

The points where the Gaussian curvature (\ref{eq:curvature}) is zero determine the critical ages of the risk function and are obtained by solving equation
\begin{equation}\label{eq:solutions-ages}
0=-(33.17 - 33.78 t + 10.35 t^2 - 0.96 t^3)^2
\end{equation}

 The solutions of (\ref{eq:solutions-ages}) are all real numbers,
 \begin{equation}
 t=1.8 \,  (5 \, \mbox{years}),   \quad t= 3.3  \, \, (26.4 \,  \mbox{years}), \quad  t=5.5  \,\, (105 \, \mbox{years})
 \end{equation}
and they correspond to the vertical lines, integral solutions of the risk vector field. This ends the proof.
 \end{proof}

\section{Discussion}\label{sec:disc}

In this paper, for the risk estimation calculated for men, using EPAs MeHg reference dose ($RfD$), and the critical region of risk (Figure \ref{fig:critic-region}),  we obtained an average value of 5.5, which is interpreted as an unacceptable health risk (since this result exceeds one). Also, there's a high probability (of $80 \%$) that some toxic or adverse health effect will develop.

 Corollary \ref{cor:risk-ages}  shows that the age of maximum risk in boys is 5 years, then the risk decreases until 26.4 years, and
 begins to increase again until reaching the maximum risk in senior men at 105 years. This later will be relevant only to those men who exceed this age.

As shown in Figure \ref{fig:curvature},  risk curvatures indicate that the highest risk is  for boys and men in reproductive age (life stages 1 and 3);  in addition, the critical region of risk is only for the aforementioned stage and to a lesser extent  for  senior men, although in this last
 group the risk curve is lower than for  boys. The study by Llop and collaborators (\cite{Llop}), recommends, infants
  and those under 3 years old avoid shark consumption; The study by Clarkson and Magos (\cite{Clarkson}) mentions that the susceptibility to neurotoxicity
  due to MeHg is related to gender, but has not been widely studied and the results available are inconclusive, but in the poisoning that occurred in
  Iraq as a consequence of the consumption of grain contaminated with a mercurial fungicide, men were affected more than men, when the
  exposure was in adulthood.

The results obtained in the work of Raimann {\it et al.} (\cite{Raimann}), coincide with the risk curves of this study where the exposure interval is higher in men mainly in reproductive age and infancy. Considering this, special care should be taken since children are more vulnerable to exposure to
methylmercury because their nervous system is the main target organ where it bio accumulates; as a precautionary measure the USEPA (\cite{USEPA1})
established an acceptable level of 0.5 mg / kg of methylmercury for fish products.

It is important to mention that  risk depends on the consumption habits (frequency of consumption and food preparation), age of the consumer,
size of the portion and the product itself. However, the magnitude of bioaccumulation of heavy metals in fish tissues is influenced by biotic
and abiotic factors, such as fish habitat, chemical form of the metal, water temperature, pH, concentration of dissolved oxygen, water transparency,
fish age, sex, body mass, and physiological conditions (Has-Schön {\it et al.}, \cite{Has-Schon}). Therefore, a more precise risk assessment will need to considere
all these factors.

\section{Conclusion}\label{sec:conclu}

The estimation of the health risk from consumption of fish substituted by shark meat, based on the results of the risk coefficient, of which an
unacceptable risk was obtained for the average and maximum MeHgconcentrations for all the population age groups, and an acceptable
risk in the low MeHgconcentration for all  age groups except for babies, for whom the risk is intermediate (0.804),
all this, means that in the analyzed sample, there is a high probability of developing deleterious health effects; so, if  men want to consume
fish products, they must buy whole fish to avoid the replacement.

The greatest uncertainty of the risk estimation in the present work, is the  lack of -direct MeHg quantification  in the
same fish samples that were genetyically analized. However,  this is an acceptable approximation for decision making in the prevention of health risks because the data used are from samplings done during a period of three years
in10 of the most important  fishing ports in Mexico, which provided a good estimate of MeHg in fish muscle sold in Mexico City's Metropolitan Area.

To analyze the uncertainties and obtain the risk function of the results obtained in the present study, a mathematical analysis was carried out using the
classical interpolation method (Reyes, \cite{Reyes}), which showed that the aforementioned risk function is stable (Golubitsky-Guillemin, \cite{Golubitsky-Guillemin}); so any error obtaining the data (uncertainties), will lead us to a risk correlation with the same characteristics (similar results), and in this way we
can conclude that the results of the risk coefficient have a high degree of reliability.

This study analysed the consumption habits of a sample of the population of Mexico City’s Metropolitan Area, which showed that with the
 substitution of 60.3 $\%$ of  fish meat for shark meat, the overall risk  was 5.55 for men, this exceeds one and it
 could be inferred that men are chronically exposed despite the fact that the population does not frequently consume fish; even more, it
 implies a health risk  for the consumer, so it is suggested to restrict the consumption of fish products to smaller rations, in lower
 frequency, and more importantly to buy complete fish to facilitate its identification of the product and to avoid consuming shark meat with larger
 concentrations of methylmercury.

\end{document}